\documentclass[pra,11pt,reprint]{revtex4-1}

\usepackage{hyperref}
\usepackage{fullpage}
\usepackage{amssymb}
\usepackage{amsthm}
\usepackage{mathtools}
\usepackage{cleveref}
\usepackage{amsfonts,mathrsfs,mathpazo,xspace,hyperref}
\usepackage{endnotes}
\usepackage{color}
\usepackage{bm}
\usepackage{times}
\usepackage{latexsym}
\usepackage{enumerate}   
\usepackage{bibunits}

\newcommand{\Mod}[1]{\ (\mathrm{mod}\ #1)}
\newcommand{\mytitle}{Learning with Errors is easy with quantum samples}
\newcommand{\Alex}{Alex B. Grilo}
\newcommand{\Iordanis}{Iordanis Kerenidis}
\newcommand{\Timo}{Timo Zijlstra}

\newtheorem{theorem}{Theorem}[section]
\newtheorem{lemma}[theorem]{Lemma}

\newtheorem{corollary}[theorem]{Corollary}

\newtheorem{definition}[theorem]{Definition}

\defaultbibliographystyle{halpha}

\newcommand {\abs} [1] {\ensuremath{ \left| #1 \right| }}
\newcommand {\norm} [1] {\ensuremath{ \left\| #1 \right\| }}
\newcommand{\eeq}{\end{eqnarray}}
\newcommand{\half}{\ensuremath{\frac{1}{2}}}

\newcommand{\ket}[1]{\ensuremath{\left|#1\right\rangle}}
\newcommand{\bra}[1]{\left\langle#1\right|}
\newcommand{\kb}[1]{\ensuremath{\ket{#1}\bra{#1}}}

\newcommand{\pr}[1]{Pr\left[#1\right]}
\newcommand{\evalue}[1]{\mathbb{E}\left[#1\right]}

\newcommand{\C}{\ensuremath{\mathbb{C}}}

\newcommand{\F}{\ensuremath{\mathbb{F}}}

\newcommand{\R}{\ensuremath{\mathbb{R}}}
\newcommand{\Z}{\ensuremath{\mathbb{Z}}}

\newcommand{\Fq}{\ensuremath{\mathbb{F}}_q}
\newcommand{\Ftwo}{\ensuremath{\mathbb{F}}_2}

\newcommand{\eps}{\varepsilon}

\newcommand{\gate}[1]{{\sf #1}}
\newcommand{\gT}[1]{\gate{T}}

\newcommand{\comment}[1]{}

\newcommand{\bounddistribution}{ k}

\newcommand{\inversebounddistribution}{\frac{\gamma q}{\bounddistribution}}
\newcommand{\goodjn}{ \sum_{\substack{j^* \in \Fq^*  \\ j^* \leq
\inversebounddistribution}}}

\newcommand{\bigslant}[2]{{\raisebox{.2em}{$#1$}\left/\raisebox{-.2em}{$#2$}\right.}}

\usepackage{times}
 \begin{document}
\title{\mytitle}
  \author{\Alex$^1$}
  \author{\Iordanis$^1$}
	\author{\Timo$^2$}
  \affiliation{$^1$ IRIF, CNRS, Universit\'{e} Paris Diderot, Paris, France}
  \affiliation{$^2$ Lab-STICC, Universit\'{e} Bretagne Sud}

  \begin{abstract}
    Learning with Errors is one of the fundamental problems in computational
    learning theory and has in the last years become the cornerstone of
    post-quantum cryptography. In this work, we study the quantum sample
    complexity of Learning with Errors and show that there exists an efficient
    quantum learning algorithm (with polynomial sample and time complexity) for the Learning with Errors problem where
    the error distribution is the one used in cryptography. While our quantum
    learning algorithm does not break the LWE-based encryption schemes proposed
    in the cryptography literature, it does have some interesting implications
    for cryptography: first, when building an LWE-based scheme, one needs to be
    careful about the access to the public-key generation algorithm that is given
    to the adversary; second, our algorithm shows a possible way for attacking
    LWE-based encryption by using classical samples to approximate the
    quantum sample state, since then using our quantum learning algorithm would
    solve LWE. Finally, we extend our results and show quantum learning algorithms for three related problems: Learning Parity with Noise, Learning with Rounding and Short Integer Solution.
  \end{abstract}

  \maketitle

  \begin{bibunit}[plain]
  \section{Introduction}
  The large amount of data arising in the real world, for example through
  scientific observations, large-scale experiments, internet traffic, social
  media, etc, makes it necessary to be able to predict some general properties
  or behaviors of the data from a limited number of samples of the data. In this
  context, Computational Learning Theory provides rigorous models for learning
  and studies  the necessary and sufficient resources, for example, the number
  of samples or the running time of the learning algorithm. In his seminal work,
  Valiant \cite{Valiant84} introduced the model of PAC
  learning, and since then this model has been extensively studied and has given rise to numerous extensions.

In another revolutionary direction, Quantum Computing takes advantage of the quantum nature of small-scale systems as a computational resource. In this
  field, the main question is to understand what problems can be solved more efficiently in a quantum computer than in
  classical computers. In the intersection of the two fields, we have Quantum Learning Theory, where we ask if quantum learning algorithms can be more efficient than classical ones. 

One of course needs to be careful about defining quantum learning and more precisely, what kind of access to the data a quantum learning algorithm has. On one hand, we can just provide classical samples to the quantum learning algorithm that can then use the quantum power in processing these classical data. In the more general scenario, we allow the quantum learning algorithm to receive quantum samples of the data, for a natural notion of a quantum sample as a superposition that corresponds to the classical sample distribution. 

More precisely, in classical learning, the learning algorithm is provided with samples of $(x, f(x))$,
  where $x$ is drawn from some unknown distribution $D$ and $f$ is the function we wish to learn. The goal of the learner
  in this case is to output a function $g$ such that with high probability (with
  respect to the samples received), $f$ and $g$ are close, i.e., $\pr{f(x) \ne
  g(x)}$ is small when $x$ is drawn from the same distribution $D$.  

The extension of this model to the quantum setting is that the samples now are
  given in the form of a quantum state $\sum_{x} \sqrt{D(x)}\ket{x}\ket{f(x)}$.
  Note that one thing the quantum learner can do with this state is simply
  measure it in the computational basis and get a classical sample from the distribution $D$. Hence, a
  quantum sample is at least as powerful as a classical sample. The main
  question is whether the quantum learner can make better use of these quantum samples and provide an advantage in the number of samples and/or running time compared to a classical learner.

In this work we focus on one of the fundamental problems in learning theory, the Learning with Errors (LWE).
  In LWE,  one is given samples of the form
  \[
  \left(a, a\cdot s + e \pmod{q} \right) 
  \]
  where $s \in \Fq^n$ is fixed, $a \in \Fq^n$ is drawn uniformly at random and
  $e \in \Fq$ is an 'error' term drawn from some distribution $\chi$. The goal is to
  output $s$, while minimizing the number of samples used and the computation
  time.

First, LWE is the natural generalisation of the well-studied Learning Parity with Noise problem (LPN), which is the case of $q=2$. Moreover, a lot
  of attention was drawn to this problem when Regev \cite{Regev05} reduced
  some (expected to be) hard problems involving lattices to LWE. With this
  reduction, LWE has become the cornerstone of current post-quantum
  cryptographic schemes. 
  Several cryptographic primitives proposals such as Fully Homomorphic
  Encryption \cite{BrakerskiV14}, Oblivious Transfer \cite{PeikertVW08}, Identity based
  encryption \cite{GentryPV08, CashHKP12,AgrawalBB10},
  and others schemes are based in the hardness of LWE (for a more complete list
  see Ref. \cite{MicciancioRegev} and Ref. \cite{Peikert16}).

  Classically, Blum et al. \cite{Blum2003} proposed the first sub-exponential algorithm for
  this problem, where both sample and time complexities are  $2^{O(n/\log{n})}$.
  Then, Arora and Ge \cite{AroraG11} improved the time complexity for LWE with a learning
  algorithm that runs in $2^{\tilde{O}(n^{2\eps})}$ time, for some $\eps <
  \frac{1}{2}$, and it uses at least $\Omega(q^2\log{q})$ 
  samples
  For LPN, 
  Lyubashevsky \cite{Lyubashevsky05} has proposed an algorithm with sample complexity $n^{1+\eps}$  at the
  cost of increasing computation time to $O(2^{n/\log{\log{n}}})$.

\section{The quantum learning model}
  
  In this work, we use the model of learning under the uniform distribution where the learner receives 
  samples according to the uniform distribution and outputs the exact function
  with  high probability.
  In the quantum setting, the learning algorithm is given quantum samples, namely a
  uniform superposition of the inputs and function values, 
  \[\sum_{x \in X} \frac{1}{\sqrt{|X|}}\ket{x}\ket{f(x)}.\]
  
In this work, we are interested in noisy samples, that can be modeled by setting
$f(x) = g(x) + e(x,r)$, where  $g$ and $e$ are deterministic functions, $x \in X$ and $r \in R$ is the randomness necessary to generate the noise. For defining the quantum sample, we start with the superposition
  \[\frac{1}{\sqrt{|R|}} \sum_{r \in R} \ket{r}\left(\frac{1}{\sqrt{|X|}} \sum_{x \in X} 
  \ket{x}\ket{g(x) + e(x,r)}\right),\]
  and then the register corresponding to the randomness is traced out. It means that with probability $\frac{1}{|R|}$ the quantum sample is  
  \[\frac{1}{\sqrt{|X|}}\ket{x}\ket{g(x) + e(x,r)}, \]
  for each possible value $r \in R$.

We consider the noise model defined in Bshouty and Jackson \cite{Bshouty95}, where
 independent noise is added for each element in the superposition, in other words, $r = (r_1,...,r_{|X|})$ and $e(x,r) = e'(r_x)$. This model is a natural generalisation for quantum samples with noise since it can be
  seen as a superposition of the classical samples. 
 
 In contrast, Cross et al. \cite{CrossSS15} proposed a noise function that is independent of $x$. 
Although our noise model might require exponentially more resources to implement quantum samples, we show that this does not make the problem intractable. Also, this
  is the kind of state we would get after solving the index erasure problem.

  \section{Our contributions}
  In this work we study quantum algorithms for solving LWE with quantum samples. Let us be more explicit on the definition of a quantum sample for the LWE problem. We assume that the quantum learning algorithm receives samples in the form

  \begin{equation}
    \frac{1}{\sqrt{q^n}} \sum_{a \in \Fq^n} \ket{a}\ket{a \cdot s
    + e_a \pmod{q}} ,
  \end{equation}
    where $e_a$ are iid random variables from some distribution
    $\chi$ over $\Fq$. 

As expected, the performance of the learning algorithm, both in the classical and quantum case, is
  sensitive to the noise model adopted, i.e. to the distribution $\chi$. When LWE is used in cryptographic schemes, the distribution $\chi$ has support on a small interval around 0, either uniform or a discrete gaussian. We prove that for such distributions, there exists an efficient quantum learner for LWE.\\

\noindent
{\bf Main Result}[informal]
{\em For error distributions $\chi$ used in cryptographic schemes, and for any
$\eta>0$, there exists a quantum learning algorithm that solves LWE with probability $1-\eta$ using $O(n \log \frac{1}{\eta})$ samples and running time $poly(n,\log \frac{1}{\eta})$.}\\

Another interesting feature of our quantum learner is that it is conceptually a
very simple algorithm based on one of the basic quantum operations, the Quantum
Fourier Transform. Such algorithms have even started to be implemented, of
course for very small input sizes and for the binary case \cite{Riste15}. Nevertheless, as far as quantum algorithms are concerned, our learner is quite feasible from an implementation point of view.

  The approach to solve the problem is a generalisation of
  Bernstein-Vazirani algorithm \cite{BernsteinV1997}: we start with a quantum sample,
  apply a Quantum Fourier Transform over $\Fq$ on each qudit, and then, we measure in the computational
  basis. Our analysis shows that, when the last qudit is not 0, which happens with high probability, the value of the remaining registers gives $s$ with constant probability. We can then repeat this process so that our algorithm outputs $s$ with high probability.

  We also use the same technique in quantum learning algorithms for three related problems.
  First, we generalise the result proposed by Cross et al.~\cite{CrossSS15} and
  Rist\`{e} et al.~\cite{Riste15} for
  the LPN problem. The main difference with their work is that we start with a
  quantum sample, i.e. a state where the noise is independent for each element
  in the superposition. Second, we show how to solve the Learning with Rounding
  problem, which can be seen as a derandomized version of LWE. Finally, we also
  propose a quantum learning algorithm for another relevant problem in
  cryptography, the Short Integer Solution problem.

  \subsection{Related work}
  We now review some results on quantum algorithms for learning problems. For a
  more extended introduction, see the survey by Arunachalam and de Wolf \cite{ArunachalamW17}.

  The first approach on trying to solve learning problems with quantum samples
  was proposed by  Bshouty and Jackson \cite{Bshouty95}, where they prove that DNFs can be learned
  efficiently, even when the samples are noisy.  No such efficient learners are
  known classically.
 
  Despite not presenting it as a learning problem, Bernstein and Vazirani \cite{BernsteinV1997}
  show how to learn parity using a single quantum sample, while classically we
  need a linear number of samples.

  Some years later, Servedio and Gortler \cite{ServedioG04} showed that classical
  and quantum sample/query complexity of learning problems are polynomially
  related, but they showed that for time complexity there exist exponential
  separations between classical and quantum learning (assuming standard
  computational hardness assumptions).

  Then, Ambainis et al. \cite{Ambainis2004}, Atici and Servedio \cite{AticiS05},
  and Hunziker et al. \cite{Hunziker2010} provided general upper
  bounds on the query complexity for learning problems that depend on the size
  of the concept class being learned.

  On specific problems, Atici and Servedio \cite{AticiS07} and Belovs
  \cite{Belovs2015} provided quantum algorithms for learning juntas and Cross et
  al. \cite{CrossSS15} proposed and implemented quantum algorithms for
  LPN in a different noise model.
  
  Recently, Arunachalam and de Wolf \cite{ArunachalamW16} proved optimal bounds
  for the quantum sample complexity of the Quantum PAC model.

 \subsection{Relation to LWE-based cryptography}
 
 As we have mentioned, LWE is used in cryptography for many different tasks. Let us briefly describe how one can build an encryption scheme based on LWE \cite{Regev05}. 
 The key generation algorithm produces a secret key $s \in \Fq$, while the public key
  consists of a sequence of classical LWE samples $(a_1, a_1 \cdot s + e_1
  \pmod{q}), ..., (a_m,
  a_m \cdot s + e_m  \pmod{q} )$, where the error comes from a distribution with support in a small interval around 0. For the encryption of a bit $b$, the party picks a
  subset $S$ of $[m]$ uniformly at random and outputs 
  \[\left(\sum_{i \in S} a_i  \Mod{q} , b\left\lceil\frac{q}{2}\right\rceil +  \sum_{i \in S} a_i \cdot s +
  e_i  \Mod{q} \right).\]
  For the decryption, knowing $s$ allows one to find $b$. The security analysis of the encryption scheme postulates that if an adversary can break the encryption efficiently then he is also able to solve the LWE problem efficiently. 

The algorithm we present here does not break the above LWE-based
encryption scheme. Nevertheless, it has interesting implications for cryptography.

First, our algorithm shows a possible way for attacking LWE-based encryption: use classical samples to approximate the quantum sample state, and then use our algorithm to solve LWE. One potential way for this would be to start with $m$ classical samples and create the following superposition
  \[\sum_{S \subseteq [m]} \ket{S}
  \ket{\sum_{i \in S} a_i  \Mod{q} }\ket{\sum_{i \in S} a_i \cdot s +
  e_i  \Mod{q} }.
  \]
  This operation is in fact efficient. Then, in order to approximate the quantum sample state, one would need to 'forget' the first register that contains the index information about which subset of the $m$ classical samples we took. In the most general case, such an operation of forgetting the index of the states in a quantum superposition, known as index-erasure
  (see Aharonov and Ta-Shma \cite{AharonovTS2003} and Ambainis et al. \cite{Ambainis2004}), is exponentially hard,
  and a number of problems, such as Graph Non-isomorphism, would have an
  efficient quantum algorithm, if we could do it efficiently. Nevertheless, one
  may try to use the extra structure of the LWE problem to find sub-exponential algorithms for this case.
  
A second concern that our algorithm raises is that  when building an LWE-based scheme, one needs to be careful on the access to the public-key generation algorithm that is given to the adversary. It is well-known that for example, even in the classical case, if the adversary can ask classical queries to the LWE oracle, then he can easily break the scheme: by asking the same query many times one can basically average out the noise and find the secret $s$. However, if we just assume that the public key is given as a box that an agent has passive access to it, in the sense that he can request a random sample and receive one, then the encryption scheme is secure classically as long as LWE is difficult. However, imagine that the random sample from LWE is provided
  by a device that creates a superposition 
    $\frac{1}{\sqrt{q^n}} \sum_{a \in \Fq^n} \ket{a}\ket{a \cdot s
    + e_a \pmod{q}}$ and then measures it. 
Then a quantum adversary that has access to this quantum state can break the scheme. Again, our claim is, by no means, that our algorithm breaks the proposed LWE-based encryption schemes, but more that LWE-based schemes which are secure classically (assuming the hardness of LWE) may stop being secure against quantum adversaries if the access to the public key generation algorithm becomes also quantum. 

  A similar situation has also appeared in the symmetric key cryptography with the so called 
  superposition attacks \cite{ Zhandry12, Boneh2013,Damgard2014, Kaplan2016}.
  There, the attacker has the ability to query the encryption oracle in
  superposition, and in this way, she can in fact break many schemes that are assumed to be secure classically. 
While in the case of symmetric cryptography, the attacker must have quantum access to the encryption oracle in order to break the system, our results show that in the case of LWE-based public-key encryption, the attacker must have quantum access to the public key generation algorithm.

\section{Algorithm for LWE}\label{sec:algo-lwe}
In this section we present the extension of the Bernstein-Vazirani algorithm for
higher order fields and analyse its behaviour with LWE samples.

We show now the Field Bernstein-Vazirani algorithm \cite{BernsteinV1997} and its main component is the  Quantum Fourier Transform over $\Fq$,
  $    QFT\ket{j} = \frac{1}{\sqrt{q^{n}}} \sum_{k = 0}^{q^n - 1}
  \omega^{jk}\ket{k}$.

\noindent \rule[1ex]{0.45\textwidth}{0.5pt}

\noindent
{\bf Field Bernstein-Vazirani algorithm}\\
{\bf Input:} $\ket{\psi} \in (\C^2)^{\otimes n+1}$\\
{\bf Output:} $\tilde{s} \in \Fq^n \cup \{\perp\}$

\noindent Apply $QFT^{\otimes n+1}$  on \ket{\psi}.

\noindent Measure in the computational basis

\noindent Let $\ket{j}\ket{j^*}$ be the output

\noindent If\;  $j^* \ne 0$, return $ -(j^*)^{-1}j \pmod{q}$\\
      Else, return $\perp$

\noindent \rule[1ex]{0.45\textwidth}{0.5pt}

For warming-up, we show the behaviour of Field Bernstein-Vazirani for learning
linear functions without noise in \Cref{sec:learning-parity} and then in
\Cref{sec:noisy} we analyse it for LWE samples. 

  \subsection{Quantum algorithm for learning a linear function without error} \label{sec:learning-parity}
If the input $\ket{\psi}$ is a noiseless quantum sample of a linear function, 
namely
  \begin{align}
    \label{eq:sample}
  \ket{\psi} = \frac{1}{\sqrt{q^n}} \sum_{a \in \Fq^n}
    \ket{a}\ket{a \cdot s \pmod{q}},
  \end{align}
Then the Field-Bernstein Vazirani outputs the correct value with probability
$\frac{q-1}{q}$: after applying the QFT on each qudit of \cref{eq:sample},  we get
  the state
  \begin{align*}
    \frac{1}{q^{n+\frac{1}{2}}}\sum_{a,j \in
    \Fq^n}\sum_{j^*\in\Fq} 
    \omega^{a \cdot (j + j^*s)}
    \ket{j}\ket{j^*}.
  \end{align*}
  It is not hard to see that the probability that for all $i \in [n]$, we have
  $j = -j^*s \pmod{q}$ and $j^* \ne 0$ is
    \begin{align*}
      &\norm{\frac{1}{q^{n+\half}}\sum_{j^* \in \Fq^*}\sum_{a \in \Fq^n} 
      \omega^{0}\ket{-j^*s \pmod{q}}}^2
      &=\frac{q-1}{q}.
    \end{align*}

  Therefore, if $j^* \ne 0$, we can retrieve $s$ by outputting
  $-(j^*)^{-1} j_i$ (all operations mod q). 
 
 \subsection{Analysis of the algorithm for noisy samples}
 \label{sec:noisy}

In this section we show that the Field Bernstein-Vazirani algorithm works even
if the input is noisy.  Instead of the superposition of all elements in $\Fq^n$,
we prove our result here for a more general case
where the quantum sample has the form
\[\ket{\psi} = \frac{1}{\sqrt{v}} \sum_{a \in V}
    \ket{a}\ket{a\cdot s + e_a (mod \; q)},\] 
where $v \in [q^n]$ is a fixed value, $V$ be a random subset of $\Fq^n$ of size
$v$ and $e_a$ is a random noise. In this case, for every quantum sample, a new subset $V$ of size $v$ is picked independently at random.

\begin{theorem} \label{thm:general}
Fix $v \in [q^n]$.  Let $V \subseteq \Fq^n$ be a random subset of $\Fq^n$ such that $|V| = v$, and let 
    \[\ket{\psi} = \frac{1}{\sqrt{v}} \sum_{a \in V}
    \ket{a}\ket{a\cdot s + e_a (mod \; q)},\] 
    where the $e_a$ are random variables with absolute value at most $k$.
The Field Bernstein-Vazirani$(\ket{\psi})$ outputs $s$ with probability
$\frac{v}{20kq^{n}}$.
\end{theorem}
  \begin{proof}
     If we apply QFT on the state \ket{\psi}, we have
  \begin{align*}
    \frac{1}{\sqrt{q^{n+1}v}}\sum_{a \in
    V}\sum_{j\in\Fq^n, j^* \in \Fq} 
    \omega^{e_a j^* + a \cdot (j + j^*s)}
    \ket{j}\ket{j^*}.
  \end{align*}

From the last equation, we have that the probability that $j =
    -j^*s \pmod{q}$ and $j^* \ne 0$ is:
     \begin{align}
       &
      \frac{1}{q^{n+1}v}\norm{
\sum_{a \in
    V}\sum_{j^* \in \Fq^*} 
    \omega^{e_a j^*}
    \ket{-j^*s \pmod{q}}\ket{j^*}
       }^2  \nonumber \\ 
      &=
      \frac{1}{q^{n+1}v}\sum_{j^* \in \Fq^*}
      \left( \sum_{a \in V}
       \mathfrak{R}(\omega^{e_a j^*})
       \right)^2
      + \left( \sum_{a \in V}
       \mathfrak{I}(\omega^{e_a j^*})
       \right)^2
      \nonumber\\      
      &\geq
      \frac{1}{q^{n+1}v}
      \goodjn \left(
       \sum_{a \in V}
       \mathfrak{R}(\omega^{e_a j^*})
       \right)^2 \nonumber\\
      &\geq
       \frac{\gamma v\cos\left(2\pi\gamma\right)^2}{kq^{n}}\nonumber. 
    \end{align}
    where $\gamma \in (0, \frac{1}{4})$ $\mathfrak{R}(z)$ and $\mathfrak{I}(z)$ are
    the real and imaginary part of $z$, respectively. For the first inequality, 
    we have removed some positive
    quantities, and the last inequality follows from the fact that
    $\mathfrak{R}(\omega^{e_a j^*}) \leq \cos\left(2\pi\gamma\right)$
for $j^* \leq \inversebounddistribution$ and 
$\abs{e_a} \leq
    \bounddistribution$. The result follows by maximizing the quantity over all
    $\gamma \in (0, \frac{1}{4})$.
\end{proof}
We now propose an algorithm that tests a candidate solution.

~\\
\rule[1ex]{0.45\textwidth}{0.5pt}

\noindent{\bf Test Candidate}\\
{\bf Input:} $\tilde{s} \in \Fq^n$, $M \in \Z^+$\\
{\bf Output:} Accept/reject

\noindent Repeat $M$ times.

	Pick sample  $\ket{\psi} =  \frac{1}{\sqrt{v}} \sum_{a \in V} \ket{a}\ket{a \cdot s + e_a \mod q}$.	

  Measure the sample in the computational basis
  
  Let $(a', a'\cdot s + e_{a'})$ be the output

   If $|a'\cdot s + e_{a'} - a'\cdot \tilde{s}| > k$, reject

\noindent Accept

\noindent \rule[1ex]{0.45\textwidth}{0.5pt}

\begin{lemma} \label{lemma:test}
For $\tilde{s} = s$, Test Candidate$(\tilde{s}, M)$ accepts with probability $1$,  while for 
    $\tilde{s} \ne s$, Test Candidate$(\tilde{s}, M)$ accepts 
    with probability probability at most $\left(\frac{2k+1}{q}\right)^M$. 
\end{lemma}
\begin{proof}
  Since
   $|a'\cdot s + e_{a'} - a'\cdot s| = |e_{a'}| \leq k$ by the noise
   distribution, it follows that the
   test passes with probability $1$ when $\tilde{s} = s$.

For a value $a'$ picked uniformly random from  $\Fq^n$, it follows that $a' \cdot (s - \tilde{s}) + e_{a'} \mod q$
  is uniformly distributed over $\F_q$
if $\tilde{s} \neq s$. Therefore, the probability that 
  it lies in the interval $[-k, k]$ is $\frac{2k+1}{q}$. Since the
  probability is independent for every iteration, the probability that
  $\tilde{s}$ is accepted on $M$ iterations is 
  $\left(\frac{2k+1}{q}\right)^M$.
\end{proof}

In \Cref{sec:lwe-theorem}, we show how to use the previous algorithms to achieve
the following theorem.

 \begin{theorem}
   \label{thm:crypto_lwe}
    For dimension $n$, let $q$ be a prime in the interval $[2^{n^\gamma},
    2\cdot 2^{n^\gamma})$.  Let 
    \[\ket{\psi} = \frac{1}{\sqrt{q^n}} \sum_{a \in \Fq^n}
    \ket{a}\ket{a \cdot s + e_a},\] 
    where the $e_a$ are random variables drawn from a noise
    distribution with noise magnitude at most 
    $k = poly(n)$. There is an algorithm that outputs $s$ with probability $1 - \eta$ with sample complexity $O(k\log{\frac{1}{\eta}})$ and running time $poly(n, \log \frac{1}{\eta})$.
 \end{theorem}

We show then how to extend the result to related problems in
\Cref{sec:related-problems}.

\section{Open problems}

\subsection{Generalizing from linear functions}
Learning linear functions can be seen as finding a hidden subgroup $H = {a | a
\cdot s = 0}$ of $\Z_q^n$. Efficient algorithms for general Abelian Hidden Subgroup Problem are known \cite{CheungM01}\cite{Mosca2014}, and we leave as an open question if these algorithms are also tolerant to noise.

\subsection{LWE over rings} \label{sec:open-problems-ring}
Due to technical reasons regarding the representation of polynomials in Ring-LWE
instances (see  \Cref{sec:rings} for more details), our LWE algorithm cannot be used to solve Ring-LWE with quantum samples and we leave this question as an open problem.

\section*{Acknowledgments}

AG and IK thank Ronald de Wolf for helpful discussions. AG thanks also Lucas
Boczkowski, Brieuc Guinard, François Le Gall and Alexandre Nolin for helpful
discussions. Supported by ERC QCC and  French Programme d’Investissement
d’Avenir RISQ P141580.

\putbib[references]
\end{bibunit}

\appendix
\begin{bibunit}[plain]
  
  \section{Notation}
  For $n \in \mathbb{N}$, we define  $[n] := \{1, ..., n\}$.
  For a complex number $x = a + ib$, $a,b \in \R$, we define its norm $|x|$ by
  $\sqrt{a^2 + b^2}$, its real part $\mathfrak{R}(x) = a$ and its imaginary part
  $\mathfrak{I}(x) = b$. We denote $\omega$ as the $q$-th root of unity, where $q$
  will be clear by the context.
  For a field $\Fq$ and element $a \in \Fq$, we denote $|a|$ as the unique value
  $b \in \left[\frac{-(q-1)}{2}, \frac{q-1}{2}\right]$ such that $b \equiv a \pmod{q}$. 

  We remind now the notation for quantum information and computation. For
  readers not familiar with these concepts we refer Ref. \cite{NielsenC2011}.  Let $\{e_i\}$ be the standard basis for the
  $q$-dimensional Hilbert space $\mathbb{C}^q$. We denote here $\ket{i} = e_i$
  and a $q$-dimensional qudit is a unit vector in this space, i.e.
  $\ket{\psi} = \sum_{i \in \Fq} \alpha_i\ket{i}$, for $\alpha_i \in \mathbb{C}$ and
  $\sum_{i \in \Fq} \abs{\alpha_i}^2 = 1$. We call the state a qubit when $q =
  2$. A $k$-qudit quantum state is
  a unit vector in the complex Hilbert space $\C^{q^k}$ and we shorthand
  the basis states for this space $\ket{i_1}\otimes...\otimes\ket{i_k}$ with
  $\ket{i_1}...\ket{i_k}$.

\section{An efficient quantum learning algorithm for LWE}
  \label{sec:lwe-theorem}
 In this section we show how to use \Cref{thm:general,lemma:test} in order to solve solve LWE
 with quantum samples using noise distributions proposed in
  Brakerski and Vaikuntanathan \cite{BrakerskiV14}, proving
  \Cref{thm:crypto_lwe}. There, the field order $q$ is sub-exponential in the
dimension $n$, generally in $[2^{n^\gamma},2\cdot 2^{n^\gamma})$ for some constant $\gamma \in (0,1)$, while the noise
distribution $\chi$ produces samples with magnitude at most polynomial in $n$
(for instance linear).

\noindent \rule[1ex]{0.45\textwidth}{0.5pt}

\noindent
{\bf LWE Algorithm(L, M)}\\
{\bf Input:} $L, M \in \Z^+$\\
{\bf Output:} $\tilde{s} \in \Fq^n \cup \{\perp\}$

\noindent Repeat $L$ times:

Pick a quantum sample $\ket{\psi}$ 

 Run the Field Bernstein-Vazirani($\ket{\psi}$) to get output $\tilde{s}$

 Run Test Candidate$(\tilde{s}, M)$

If $\tilde{s}$ passes the test, return $\tilde{s}$

\noindent Return $\perp$.

\noindent 
\rule[1ex]{0.45\textwidth}{0.5pt}

\begin{theorem}\label{thm:main}
LWE Algorithm$(L, M)$ outputs $s$ with probability
\[1 - \left(1-\frac{v}{20kq^{n}}\right)^{L} - \left(\frac{3k}{q}\right)^M L. \]
\end{theorem}
\begin{proof}
LWE Algorithm$(L, M)$ does not output $s$ if either  Test Candidate$(\tilde{s}, \log{\frac{1}{\eta}})$  accepts some $\tilde{s} \ne s$ before an iteration where Field Bernstein-Vazirani outputs  $s$, or LWE Algorithm outputs $\perp$. We can upper bound the probability of this event by the probability  that at least one of $L$ independent calls to Test Candidate$(\tilde{s}, \log{\frac{1}{\eta}})$ accepts some $\tilde{s} \ne s$ or that $L$ independent calls to Field Bernstein-Vazirani do not output $s$. 

From \Cref{lemma:test} and using the union bound, the probability that at least one of $L$ independent calls to Test Candidate$(\tilde{s}, \log{\frac{1}{\eta}})$ accepts  some $\tilde{s} \ne s$  is at most
  \[\left(\frac{2k+1}{q}\right)^M L \leq
  \left(\frac{3k}{q}\right)^M L.
  \]

From \Cref{thm:general}, the probability that $s$ is not the output of $L$ independent calls to Field Bernstein-Vazirani is at most
  \[\left(1-\frac{v}{20kq^{n}}\right)^{L}.\]

By union bound, LWE Algorithm$(L, M)$ does not output $s$ with probability at most
  \[\left(1-\frac{v}{20kq^{n}}\right)^{L} + \left(\frac{3k}{q}\right)^M L. \qedhere \]
\end{proof}

\Cref{thm:crypto_lwe} follows directly from \Cref{thm:main} and picking $v = q^n$,  $L = 20k \ln \frac{1}{\eta}$ and $M = 1$.

\section{Quantum learning complexity of related problems}
\label{sec:related-problems}
In this section we present learning algorithms for problems that are related to LWE.

\subsection{Learning parity with noise}
We show here our result for Learning Parity with Noise (LPN) problem,
which is the LWE problem for $q=2$. 

Here, the parity bit is flipped independently for
each element in the superposition with probability $\eta$. This
is the same noise model proposed by Bshouty and Jackson \cite{Bshouty95}. Note
that Cross et al.~\cite{CrossSS15} studied LPN with different noise models. In
the first, all parities in the superposition are flipped at the same time with
probability $\eta$. In the second one, each qubit passed through a depolarising
channel. Our algorithm and analysis also works for both of the noise models proposed by
Cross et al.~\cite{CrossSS15}.

The algorithm is the same as in the previous section, where now the QFT is over
$\Ftwo$ (also called the Hadamard Transform H). 

\begin{lemma}
    \label{lem:parity}
    Let 
    $\frac{1}{\sqrt{2^n}} \sum_{a \in \{0,1\}^n}
    \ket{a}\ket{a \cdot s  + e_a \pmod{2}}$
    be a quantum sample where $e_a$ are iid random variables with value $0$ with
    probability $1-\eta$ and $1$ with probability $\eta$.

    For every constant $0 < \delta < 1$, applying a Hadamard transform on all qubits and measuring them in the computational
     basis, provides an outcome $\ket{j}\ket{j^*}$, where $j \in \{0,1\}^n$ and $j^* \in \{0,1\}$ such that 
     with probability exponentially close to $1$, $\pr{j = s} \geq
      \frac{1}{2} (1-\delta)^2(1-2\eta)^2$.
\end{lemma}
  \begin{proof}
  If we apply Hadamards on each qubit of the sample state, we have
  \begin{align*}
    \frac{1}{2^{n+\frac{1}{2}}}\sum_{a \in \{0,1\}^n}\sum_{j\in\{0,1\}^{n+1}} 
    (-1)^{e_aj^* + a \cdot(j + j^*s)}
    \ket{j}\ket{j^*} 
  \end{align*}

   We now calculate the probability 
    that $j^* = 1$ and the first qubits are in the state $\ket{s}$: 

    \begin{align*}
      &\norm{\frac{1}{2^{n+\half}}\sum_{a \in \{0,1\}^n} 
      (-1)^{e_a  + a \cdot (s + s)}\ket{s}}^2
      \\
      &= \frac{1}{2^{2n}+1}
      \left(\sum_{a \in \{0,1\}^n}  (-1)^{e_a}\right)^2
    \end{align*}

  From the distribution of each $e_a$, we have that 
    $(-1)^{e_{a}}$ is $1$ w.p. $1-\eta$ and $-1$ w.p. $\eta$, independently. Therefore $\evalue{(-1)^{e_a}} = 1-2\eta$
  and using Hoeffding's bound  we have that
    \begin{align*}
      &\pr{\sum_{a \in \{0,1\}^{n}} (-1)^{e_{a}} \leq (1-\delta)(1-2\eta)2^{n}}
      \\&<e^{\delta^2 (1-2\eta)^2 2^{2n}/4} 
    \end{align*}

  Therefore, with probability exponentially close to $1$, the probability that
  $j = s$ is at least
  \[
    \frac{1}{2^{2n}+1} ((1-\delta)(1-2\eta)2^{n})^2 = 
    \frac{1}{2}(1-\delta)^2(1-2\eta)^2. \qedhere
  \]
\end{proof}

We can test a candidate solution $\tilde{s}$ in the lines of \Cref{lemma:test} and then repeat
the process a linear number of times, and in this case the algorithm can find
the right $s$ with probability exponentially close to $1$.

\subsection{LWE over rings}
\label{sec:rings}

The Ring-LWE problem\cite{LPR10}, a variant of LWE over the ring of polynomials, has been proposed in order to improve the performance of cryptographic constructions using LWE, at the cost of needing
stronger assumptions for proving its hardness.

The Ring-LWE problem uses the structure of the ring $\mathcal{R}_q = \bigslant{\mathcal{R}}{q\mathcal{R}}$ for a prime $q$, $\mathcal{R} = \bigslant{\mathbb{Z}[x]}{f(x)}$ and a cyclotomic polynomial $f(x)$. 
As in LWE, a Ring-LWE sample is the pair $(a, as + e \pmod{q})$ for random $s,a
\in \mathcal{R}_q$ and $e$ is picked according to some error distribution $\chi$.

Unfortunately, our algorithm cannot be used to solve Ring-LWE with the noise model proposed by Bshouty et al.~\cite{Bshouty95}, due to technical issues on representing the polynomials.
In order to use the quantum learning algorithm for LWE, we need to find an isomorphism $\phi$ from $\mathcal{R}_q$ to $\mathbb{Z}_q^n$, where $n = \varphi(m)$ is the number of invertible elements modulo $m$.
With this isomorphism, we can consider a sample $(a, as+e) \in \mathcal{R}_q^2$ as two vectors in $\mathbb{Z}_q^n$, and a superposition of quantum states representing these vectors can be written as:
\begin{equation*} 
\ket{\psi} = \frac{1}{\sqrt{q^n}} \sum_{a \in \mathcal{R}_q} \ket{\phi(a)}\ket{\phi(as + e_a)},
\end{equation*}
and applying the QFT over every register of this state results in 
\begin{align} \label{eq:qftring}
  &QFT^{\otimes 2n} \ket{\psi} \\
  &= \frac{1}{\sqrt{q^{3n}}} \sum_{a \in \mathcal{R}_q} \sum_{x,y \in \mathbb{Z}_q} \omega^{\phi(a) \cdot x} \ket{x} \otimes \omega^{\phi(as + e_a) \cdot y} \ket{y} \\
&= \frac{1}{\sqrt{q^{3n}}} \sum_{a \in \mathcal{R}_q} \sum_{x,y \in \mathbb{Z}_q} \omega^{\phi(a)\cdot(x+y\phi(s)) + y\cdot\phi(e_a)} \ket{x}\ket{y}, \nonumber
\end{align}
where the second equality holds because $\phi$ is a homomorphism.

We consider two ways of representing elements in $\mathcal{R}_q$ as integer vectors.
The first one consists of identifying a polynomial in $\mathcal{R}_q$ with the vector containing its coefficients. However, this 
coefficient embedding is not a homomorphism to $\Z_q^n$, and the following identity, used in \cref{eq:qftring}, does not hold
\begin{equation*}
\phi(a)\cdot x + \phi(a\cdot s + e_a)y = \phi(a)(x + y\phi(s)) + y \phi(e_a).
\end{equation*}

Therefore, this representation of polynomials cannot be used within our learning algorithm.

The second way of representing a polynomial is through the map 
\begin{align*}
\phi(p(x)) = (p(\omega_m), \dots, p(\omega_m^{m-1})),
\end{align*}
where  $\omega_m \in \mathbb{Z}_q$ be a primitive $m$-th root of unity.
This map is particularly interesting since multiplication  $\mathbb{Z}_q^n$ is done component-wise\cite{LPR13}
and therefore it can be used in implementations of Ring-LWE with efficient multiplication\cite{Mayer16}.
However, in these constructions, the error is sampled from a distribution over
polynomials with small coefficients and when after applying the isomorphism,
$\phi(e_a)$ can be arbitrarily large in $\mathbb{Z}_q^n$ which cannot be handled
by our algorithm if the error is independent for each element in the
superposition.

~\\
Finally, we show now how to do solve Ring-LWE for the error model presented in Cross et al.~\cite{CrossSS15}, namely, the noise is the same for all elements in the superposition.

Let $\phi$ be any isomorphism from $\mathcal{R}_q$ to $\mathbb{Z}_q^n$. We can map the original quantum sample using $\phi$ resulting in 
\begin{align*}
\frac{1}{\sqrt{q^{n}}} \sum_{a \in \mathcal{R}_q} \ket{\phi(a)} \otimes \ket{\phi(as + e)}.
\end{align*}
and using the Field Bernstein-Vazirani algorithm on this state we have
\begin{align*}
  &QFT^{\otimes 2n} \ket{\psi} \\
  &= \frac{1}{\sqrt{q^{3n}}} \sum_{a \in \mathcal{R}_q} \sum_{x,y \in \mathbb{Z}_q} \omega^{\phi(a) \cdot x} \ket{x} \otimes \omega^{\phi(as + e) \cdot y} \ket{y} \\
&= \frac{1}{\sqrt{q^{3n}}} \sum_{y\in \Z_q}  
\omega^{y\cdot\phi(e)} 
\sum_{a \in \mathcal{R}_q} \sum_{x \in \mathbb{Z}_q} \omega^{\phi(a)\cdot(x+y\phi(s)))} \ket{x}\ket{y}.
\end{align*}
By measuring the last register, the error becomes a global phase and we are able to retrieve $s$ as shown in \Cref{sec:learning-parity}.

\subsection{Learning with Rounding}
LWE has been used in the construction of several cryptographic primitives.
However, its usage sometimes is limited. For instance, in the implementation of
pseudo-random functions, the output must use little or no randomness, which does
not correspond to the inherent randomness in
LWE's input.

For this purpose, Banerjee, Peikert and Rosen\cite{BPR11} proposed a derandomized version of LWE called Learning with Rounding (LWR), which does not compromise hardness. LWR has been used in the construction of pseudo-random functions \cite{BPR11} and deterministic public key encryption \cite{XXZ12}.

The main idea of LWR consists in replacing $a \cdot s + e_a$
by the 'rounding' of $a\cdot s$  with respect to some modulus $p\ll q$, which can be seen as a ``deterministic noise''.
More precisely, the rounding function is defined as follows:
\begin{align*}
\left\lfloor \cdot \right\rceil_p : \mathbb{Z}_q &\rightarrow \mathbb{Z}_p, \text{ with} \left\lfloor x \right\rceil_p  = \left\lfloor \frac{p}{q}x \right\rceil \pmod{p}.
\end{align*}
An LWR sample is then given by $(a, \left\lfloor a \cdot s \right \rceil_p)$ for some $a$ sampled from the uniform distribution on $\Fq^n$. 

\begin{corollary}
 Let 
    \[\ket{\psi} = \frac{1}{\sqrt{q^n}} \sum_{a \in \Fq^n}
    \ket{a}\ket{\left\lfloor a \cdot s \right \rceil_p},\] 
    be a quantum LWR sample. Let $\ket{\phi}$ be the state when we multiply the
    last register of $\ket{\psi}$ with $\frac{q}{p}$. The Field Bernstein-Vazirani$(\ket{\phi})$ outputs $s$ with probability
     at least $\frac{p}{12(q-1)}$.
\end{corollary}
\begin{proof}
For a fixed $a$,  we have that
\[\frac{q}{p}\left\lfloor a \cdot s \right \rceil_p  = a\cdot s + \left(\frac{q}{p}\left\lfloor a \cdot s \right \rceil_p - a\cdot s\right)\pmod{q} .\]
Since $\frac{-q}{2p} \leq \frac{q}{p}\left\lfloor a \cdot s \right \rceil_p - a\cdot s \leq \frac{q}{2p}\pmod{q} $, 
the result follows by 
\Cref{thm:general} for  $k = \frac{q}{2p}$.
\end{proof}
	
\subsection{Quantum Samples for SIS problem}	
We present in this section a learning algorithm for another relevant problem in cryptography, the Short Integer Solution problem.  As the name indicates, the Short Integer Solution problem (SIS) consists in finding a short integer solution for a  system of linear equations, and we present now its formal definition.

\begin{definition}[Short Integer Solution]
Given a random matrix $A \in \Fq^{m \times n}$, a random vector $z \in
  \Fq^m$, the SIS$_{n,m,q,\beta}$ problem is to find a
  vector $x \in \Fq^n$ such
  that $Ax = z \pmod{q}$ with $\left\|x\right\| < \beta$. 
\end{definition}

As in the LWE case, the hardness of SIS is also proved through the reduction of
(expected to be) hard lattice problems~\cite{Ajtai96}\cite{MicciancioR07}\cite{GentryPV08}\cite{MicciancioP13}.
We remark that if we drop either the constraint of
having an {\em integer} solution or having a {\em short} solution, the problem
can be easily solved using Gaussian Elimination. 

The SIS problem and its variants have been used to prove security of
constructions of signature schemes~\cite{Ly12}\cite{DDLL13}, and hash
functions~\cite{LMPR08}. In these schemes, samples in the form $(A, Av)$ are
public, where $v$ is a small random vector and $A$ is a random matrix.

Inspired in the LWE case,  we can define a quantum sample for SIS problem as
\[\ket{\psi} = \frac{1}{\sqrt{q^{nm}}}\sum_{A \in \Fq^{m\times n}}
\ket{A}\ket{Av} \pmod{q},\]
and we are interested in the sample complexity of finding the (fixed) short solution $v$. Using Field Bernstein-Vazirani brings the same problem of Gaussian Elimination: there is no guarantee of finding a short solution instead of an arbitrary one.

We notice that tracing out $m-1$ rows of $A$ and the corresponding positions of $Av$, we remain with
\begin{equation*}\label{eq:sample-sis}
\frac{1}{\sqrt{q^{n}}}\sum_{a \in \Fq^{n}} \ket{a}\ket{a\cdot v},
\end{equation*}
and we show an algorithm that works even for this type of quantum sample.

The algorithm consists by testing all possible values $j \in \{-k,...,k\}$ of $-v_i$. The test on $j = -v_i$ passes with probability $1$, while the test rejects with constant probability for $j \ne -v_i$. By repeating the test $L$ times, the probability of finding the correct value is amplified.

\noindent \rule[1ex]{0.45\textwidth}{0.5pt}
~\\
{\bf SIS Algorithm(L)}\\
{\bf Input:} $L \in \Z^+$\\
{\bf Output:} $\tilde{v} \in \Fq^n$

\noindent
For $i \in [n]$ do:

For $j \in \{-k,...,k\}$ do:

\hspace{1em} For $l \in [L]$:

\hspace{1.8em}  Pick a quantum sample $\frac{1}{\sqrt{q^{n}}}\sum_{a \in \Fq^{n}} \ket{a}\ket{a\cdot v}.$

\hspace{1.8em} Add $ja_i$ to the last register

\hspace{1.8em}  Apply QFT on the $i$-th qudit of $a$ and measure it

\hspace{1.8em}  Test next value of $j$ if outcome is not $\ket{0}$

\hspace{1em} Set $\tilde{v}_i = -j$ and continue with the next value of $i$.

\noindent Output $\tilde{v}$

\noindent \rule[1ex]{0.45\textwidth}{0.5pt}

\begin{theorem}
Let $v \in \Fq^n$ whose coefficients are all smaller in absolute value than some bound $k$. Given the quantum samples in the form
\begin{equation*}
\ket{\psi} = \sum_{a \in \Fq^{n}} \ket{a}\ket{a\cdot v}, 
\end{equation*}  
SIS Algorithm$(L)$ outputs $v$ with probability 
$1-\frac{2km}{q^L}$.
\end{theorem}
\begin{proof}
We start by doing the analysis of SIS algorithm for $i = 1$.
After adding $ja_1$ to the last register of the quantum sample, we have
  \begin{align*}
    &\frac{1}{\sqrt{q^{n}}}\sum_{a \in \Fq^{n}} \ket{a}\ket{a\cdot v + a_1j} \\
    &= \frac{1}{\sqrt{q^{n}}}\sum_{a_1 \in \Fq} \sum_{\overline{a} \in \Fq^{n-1}} \ket{a_1} \ket{\overline{a}} \ket{a_1 (v_1+j) + \overline{a} \cdot \overline{v}}.
  \end{align*}

If $j = -v_1$, then the previous state is the product state
\begin{equation*}
\frac{1}{\sqrt{q}}\sum_{a_1 \in \Fq} \ket{a_1} \otimes 
\frac{1}{\sqrt{q^{n-1}}} \sum_{\overline{a} \in \Fq^{n-1}} \ket{\overline{a}}\ket{\overline{a} \cdot \overline{v}}.
\end{equation*}
and since $QFT\sum_{a_1 \in \Fq} \ket{a_1} = \ket{0}$, the test passes for all $l \in [L]$.

On the other hand,  if $j \ne -v_1$, then the state is entangled, and the reduced density matrix of the first register is
\begin{equation*}
\frac{1}{q} \sum_{a_1 \in \Fq} \kb{a_1}.
\end{equation*}
In this case, after applying the QFT on the first register and measuring it, the output is $\ket{0}$ with probability $\frac{1}{q}$. Therefore,  we have $\tilde{v}_1 = -j$ if for all $L$ independent samples the  measurement outcome after the QFT is $\ket{0}$, and this happens with probability $\frac{1}{q^L}$.
By the union bound, the probability that the test passes for any value $j \ne -v_i$ is at most $\frac{2k}{q^L}$.

Finally, the previous analysis holds for every $i \in [n]$. Since  
 $v \ne \tilde{v}$ iff there exists an $i \in [n]$ such that $\tilde{v}_i \ne v_i$, we can use union bound again to show that  this happens with probability  at most $\frac{2km}{q^L}$.
\end{proof}

By picking $L = \max\{1, \frac{\log\frac{2km}{\eta}}{\log q}\}$, the algorithm outputs the correct $v$ with probability at least $1 - \eta$.

\putbib[references]
\end{bibunit}

\begin{thebibliography}{10}

\bibitem{AgrawalBB10}
Shweta Agrawal, Dan Boneh, and Xavier Boyen.
\newblock Efficient lattice {(H)IBE} in the standard model.
\newblock In {\em Advances in Cryptology - {EUROCRYPT} 2010}, pages 553--572,
  2010.

\bibitem{AharonovTS2003}
Dorit Aharonov and Amnon Ta-Shma.
\newblock Adiabatic quantum state generation and statistical zero knowledge.
\newblock In {\em Proceedings of the Thirty-fifth Annual ACM Symposium on
  Theory of Computing}, STOC '03, pages 20--29, 2003.

\bibitem{Ambainis2004}
Andris Ambainis, Kazuo Iwama, Akinori Kawachi, Hiroyuki Masuda, Raymond~H.
  Putra, and Shigeru Yamashita.
\newblock Quantum identification of boolean oracles.
\newblock In {\em 21st Annual Symposium on Theoretical Aspects of Computer
  Science, STACS 2004}, pages 105--116, 2004.

\bibitem{AroraG11}
Sanjeev Arora and Rong Ge.
\newblock New algorithms for learning in presence of errors.
\newblock In {\em Automata, Languages and Programming - 38th International
  Colloquium, {ICALP} 2011}, pages 403--415, 2011.

\bibitem{ArunachalamW16}
Srinivasan Arunachalam and Ronald de~Wolf.
\newblock Optimal quantum sample complexity of learning algorithms.
\newblock {\em CoRR}, abs/1607.00932, 2016.

\bibitem{ArunachalamW17}
Srinivasan Arunachalam and Ronald de~Wolf.
\newblock A survey of quantum learning theory.
\newblock {\em CoRR}, abs/1701.06806, 2017.

\bibitem{AticiS05}
Alp Atici and Rocco~A. Servedio.
\newblock Improved bounds on quantum learning algorithms.
\newblock {\em Quantum Information Processing}, 4(5):355--386, 2005.

\bibitem{AticiS07}
Alp Atici and Rocco~A. Servedio.
\newblock Quantum algorithms for learning and testing juntas.
\newblock {\em Quantum Information Processing}, 6(5):323--348, 2007.

\bibitem{Belovs2015}
Aleksandrs Belovs.
\newblock Quantum algorithms for learning symmetric juntas via the adversary
  bound.
\newblock {\em Comput. Complex.}, 24(2):255--293, June 2015.

\bibitem{BernsteinV1997}
Ethan Bernstein and Umesh Vazirani.
\newblock Quantum complexity theory.
\newblock {\em SIAM J. Comput.}, 26(5), October 1997.

\bibitem{Blum2003}
Avrim Blum, Adam Kalai, and Hal Wasserman.
\newblock Noise-tolerant learning, the parity problem, and the statistical
  query model.
\newblock {\em J. ACM}, 50(4):506--519, July 2003.

\bibitem{Boneh2013}
Dan Boneh and Mark Zhandry.
\newblock Secure signatures and chosen ciphertext security in a quantum
  computing world.
\newblock In {\em Advances in Cryptology -- CRYPTO 2013}, pages 361--379, 2013.

\bibitem{BrakerskiV14}
Zvika Brakerski and Vinod Vaikuntanathan.
\newblock Efficient fully homomorphic encryption from (standard) {LWE}.
\newblock {\em SIAM Journal on Computing}, 43(2):831--871, 2014.

\bibitem{Bshouty95}
Nader~H. Bshouty and Jeffrey~C. Jackson.
\newblock Learning dnf over the uniform distribution using a quantum example
  oracle.
\newblock In {\em Proceedings of the Eighth Annual Conference on Computational
  Learning Theory}, COLT '95, pages 118--127, 1995.

\bibitem{CashHKP12}
David Cash, Dennis Hofheinz, Eike Kiltz, and Chris Peikert.
\newblock Bonsai trees, or how to delegate a lattice basis.
\newblock {\em J. Cryptology}, 25(4):601--639, 2012.

\bibitem{CheungM01}
Kevin K.~H. Cheung and Michele Mosca.
\newblock Decomposing finite abelian groups.
\newblock {\em Quantum Information {\&} Computation}, 1(3):26--32, 2001.

\bibitem{CrossSS15}
Andrew~W Cross, Graeme Smith, and John~A Smolin.
\newblock Quantum learning robust against noise.
\newblock {\em Physical Review A}, 92(1):012327, 2015.

\bibitem{Damgard2014}
Ivan Damg{\aa}rd, Jakob Funder, Jesper~Buus Nielsen, and Louis Salvail.
\newblock Superposition attacks on cryptographic protocols.
\newblock In {\em Information Theoretic Security: 7th International Conference,
  ICITS 2013}, 2014.

\bibitem{GentryPV08}
Craig Gentry, Chris Peikert, and Vinod Vaikuntanathan.
\newblock Trapdoors for hard lattices and new cryptographic constructions.
\newblock In {\em Proceedings of the 40th Annual {ACM} Symposium on Theory of
  Computing, STOC 2008}, pages 197--206, 2008.

\bibitem{Hunziker2010}
Markus Hunziker, David~A. Meyer, Jihun Park, James Pommersheim, and Mitch
  Rothstein.
\newblock The geometry of quantum learning.
\newblock {\em Quantum Information Processing}, 9(3):321--341, 2010.

\bibitem{Kaplan2016}
Marc Kaplan, Ga{\"e}tan Leurent, Anthony Leverrier, and Mar{\'i}a
  Naya-Plasencia.
\newblock Breaking symmetric cryptosystems using quantum period finding.
\newblock In {\em Advances in Cryptology -- CRYPTO 2016}, 2016.

\bibitem{Lyubashevsky05}
Vadim Lyubashevsky.
\newblock The parity problem in the presence of noise, decoding random linear
  codes, and the subset sum problem.
\newblock In {\em Approximation, Randomization and Combinatorial Optimization,
  Algorithms and Techniques {APPROX-RANDOM} 2005}, pages 378--389, 2005.

\bibitem{MicciancioRegev}
Daniele Micciancio and Oded Regev.
\newblock Lattice-based cryptography.
\newblock In {\em Post Quantum Cryptography}, 2008.

\bibitem{Mosca2014}
Michele Mosca.
\newblock {\em Abelian Hidden Subgroup Problem}, pages 1--6.
\newblock Springer Berlin Heidelberg, Berlin, Heidelberg, 2014.

\bibitem{Peikert16}
Chris Peikert.
\newblock A decade of lattice cryptography.
\newblock {\em Foundations and Trends in Theoretical Computer Science},
  10(4):283--424, 2016.

\bibitem{PeikertVW08}
Chris Peikert, Vinod Vaikuntanathan, and Brent Waters.
\newblock A framework for efficient and composable oblivious transfer.
\newblock In {\em Advances in Cryptology - {CRYPTO} 2008}, pages 554--571,
  2008.

\bibitem{Regev05}
Oded Regev.
\newblock On lattices, learning with errors, random linear codes, and
  cryptography.
\newblock In {\em Proceedings of the 37th Annual {ACM} Symposium on Theory of
  Computing, STOC 2005}, pages 84--93, 2005.

\bibitem{Riste15}
Diego Rist\`{e}, Marcus~P. da~Silva, Colm~A. Ryan, Andrew~W. Cross, John~A.
  Smolin, Jay~M. Gambetta, Jerry~M. Chow, and Blake~R. Johnson.
\newblock Demonstration of quantum advantage in machine learning.
\newblock {\em CoRR}, abs/1512.0606G9, 2015.

\bibitem{ServedioG04}
Rocco~A. Servedio and Steven~J. Gortler.
\newblock Equivalences and separations between quantum and classical
  learnability.
\newblock {\em {SIAM} J. Comput.}, 33(5):1067--1092, 2004.

\bibitem{Valiant84}
L.~G. Valiant.
\newblock A theory of the learnable.
\newblock {\em Commun. ACM}, 27(11):1134--1142, 1984.

\bibitem{Zhandry12}
Mark Zhandry.
\newblock How to construct quantum random functions.
\newblock In {\em 53rd Annual {IEEE} Symposium on Foundations of Computer
  Science, {FOCS} 2012}, pages 679--687, 2012.

\end{thebibliography}


\begin{thebibliography}{10}

\bibitem{Ajtai96}
M.~Ajtai.
\newblock Generating hard instances of lattice problems (extended abstract).
\newblock In {\em Proceedings of the Twenty-eighth Annual ACM Symposium on
  Theory of Computing}, STOC '96, pages 99--108, New York, NY, USA, 1996. ACM.

\bibitem{BPR11}
Abhishek Banerjee, Chris Peikert, and Alon Rosen.
\newblock Pseudorandom functions and lattices.
\newblock Cryptology ePrint Archive, Report 2011/401, 2011.
\newblock \url{http://eprint.iacr.org/2011/401}.

\bibitem{BrakerskiV14}
Zvika Brakerski and Vinod Vaikuntanathan.
\newblock Efficient fully homomorphic encryption from (standard) {LWE}.
\newblock {\em SIAM Journal on Computing}, 43(2):831--871, 2014.

\bibitem{Bshouty95}
Nader~H. Bshouty and Jeffrey~C. Jackson.
\newblock Learning dnf over the uniform distribution using a quantum example
  oracle.
\newblock In {\em Proceedings of the Eighth Annual Conference on Computational
  Learning Theory}, COLT '95, pages 118--127, 1995.

\bibitem{CrossSS15}
Andrew~W Cross, Graeme Smith, and John~A Smolin.
\newblock Quantum learning robust against noise.
\newblock {\em Physical Review A}, 92(1):012327, 2015.

\bibitem{DDLL13}
L{\'e}o Ducas, Alain Durmus, Tancr{\`e}de Lepoint, and Vadim Lyubashevsky.
\newblock Lattice signatures and bimodal gaussians.
\newblock In {\em 33rd Annual Cryptology Conference, CRYPTO 2013}, pages
  40--56, 2013.

\bibitem{GentryPV08}
Craig Gentry, Chris Peikert, and Vinod Vaikuntanathan.
\newblock Trapdoors for hard lattices and new cryptographic constructions.
\newblock In {\em Proceedings of the 40th Annual {ACM} Symposium on Theory of
  Computing, STOC 2008}, pages 197--206, 2008.

\bibitem{Ly12}
Vadim Lyubashevsky.
\newblock {\em Lattice Signatures without Trapdoors}, pages 738--755.
\newblock Springer Berlin Heidelberg, Berlin, Heidelberg, 2012.

\bibitem{LMPR08}
Vadim Lyubashevsky, Daniele Micciancio, Chris Peikert, and Alon Rosen.
\newblock Swifft: A modest proposal for fft hashing.
\newblock In {\em Fast Software Encryption, 15th International Workshop, FSE
  2008}, volume 5086, pages 54--72, 2008.

\bibitem{LPR10}
Vadim Lyubashevsky, Chris Peikert, and Oded Regev.
\newblock On ideal lattices and learning with errors over rings.
\newblock {\em J. ACM}, 60(6):43:1--43:35, November 2013.

\bibitem{LPR13}
Vadim Lyubashevsky, Chris Peikert, and Oded Regev.
\newblock A toolkit for ring-lwe cryptography.
\newblock In {\em Advances in Cryptology - {EUROCRYPT} 2013, 32nd Annual
  International Conference on the Theory and Applications of Cryptographic
  Techniques}, pages 35--54, 2013.

\bibitem{Mayer16}
Christoph~M. Mayer.
\newblock Implementing a toolkit for ring-lwe based cryptography in arbitrary
  cyclotomic number fields.
\newblock {\em {IACR} Cryptology ePrint Archive}, 2016:49, 2016.

\bibitem{MicciancioP13}
Daniele Micciancio and Chris Peikert.
\newblock Hardness of sis and lwe with small parameters.
\newblock In {\em Advances in Cryptology -- CRYPTO 2013: 33rd Annual Cryptology
  Conference, Santa Barbara, CA, USA, August 18-22, 2013. Proceedings, Part I},
  pages 21--39, 2013.

\bibitem{MicciancioR07}
Daniele Micciancio and Oded Regev.
\newblock Worst-case to average-case reductions based on gaussian measures.
\newblock {\em SIAM J. Comput.}, 37(1):267--302, April 2007.

\bibitem{NielsenC2011}
Michael~A. Nielsen and Isaac~L. Chuang.
\newblock {\em Quantum Computation and Quantum Information: 10th Anniversary
  Edition}.
\newblock Cambridge University Press, New York, NY, USA, 10th edition, 2011.

\bibitem{XXZ12}
Xiang Xie, Rui Xue, and Rui Zhang.
\newblock Deterministic public key encryption and identity-based encryption
  from lattices in the auxiliary-input setting.
\newblock Cryptology ePrint Archive, Report 2012/463, 2012.
\newblock \url{http://eprint.iacr.org/2012/463}.

\end{thebibliography}
\end{document}